\documentclass[lettersize,journal]{IEEEtran}
\usepackage{amsmath,amsfonts}
\usepackage[linesnumbered,ruled,vlined]{algorithm2e}
\usepackage{array}
\usepackage{textcomp}
\usepackage{stfloats}
\usepackage{url}
\usepackage{verbatim}
\usepackage{graphicx}
\usepackage{subcaption}
\usepackage{cite}
\usepackage{hyperref}
\usepackage{balance}
\usepackage{amsthm}
\usepackage{xcolor}
\usepackage{enumitem}
\usepackage{pifont}

\theoremstyle{definition} 
\newtheorem{definition}{Definition} 
\newtheorem{proposition}{Proposition}
\newtheorem{theorem}{Theorem}
\newtheorem{lemma}{Lemma}
\newtheorem{assumption}{Assumption}
\newcommand{\metricname}{AVA}

\graphicspath{{figures/}{images/}{plots/}}

\begin{document}

\title{Towards Efficient Service Information Refresh in Compute First Networking: Reinforcement Learning with Joint AoI and VoI}

\author{Jianpeng Qi, 
        Chao Liu,
        Chengxiang Xu,
        Rui Wang,
        Junyu Dong,
        Yanwei Yu
\thanks{Jianpeng Qi, Chao Liu, Junyu Dong, and Yanwei Yu are with the Faculty of Information Science and Engineering, Ocean University of China. Chengxiang Xu is with Inspur (Jinan) Data Technology Co., Ltd. Rui Wang is with the School of Computer and Communication Engineering, University of Science and Technology Beijing.}
\thanks{Correspondence: Yanwei Yu.}
}

\markboth{Journal of \LaTeX\ Class Files,~Vol.~14, No.~8, August~2021}%
{Jianpeng \MakeLowercase{\textit{et al.}}: Towards Efficient Status Refresh in Compute First Networking: Reinforcement Learning with Joint AoI and VoI}


\maketitle

\begin{abstract}
Timely and efficient dissemination of service information is critical in compute-first networking systems, where user requests arrive dynamically and computing resources are constrained. In such systems, the access point (AP) plays a key role in forwarding user requests to a server based on its latest received service information. This paper considers a single-source, single-destination system and introduces an Age-and-Value-Aware (AVA) metric that jointly captures both the timeliness and the task relevance of service information. Unlike traditional freshness-based metrics, AVA explicitly incorporates variations in server-side service capacity and AP’s forwarding decisions, allowing more context-aware update evaluation. 
Building upon AVA, we propose a reinforcement learning-based update policy that learns to selectively transmit service information updates to the AP. It aims to maximize overall task success while minimizing unnecessary communications. Extensive simulations under diverse user request patterns and varying service capacities demonstrate that AVA reduces the update frequency by over 90\% on average compared to baselines, with reductions reaching 98\% in certain configurations. Crucially, this reduction is achieved without compromising the accuracy of task execution or the quality of decision making.
\end{abstract}

\section{Introduction}
In vehicular networks, vehicles act as source nodes by relaying sudden incident reports to roadside units (RSUs), which serve as destination nodes. The RSUs analyze and forecast the situation, then sending commands to nearby vehicles to facilitate safe navigation \cite{RN76}. Similarly, in Compute First Networking (CFN), real-time updates on computational resources (source nodes) are communicated to the network edge (destination nodes), enabling clients to accurately assess remote available resources and offload tasks accordingly. In unmanned aerial vehicle (UAV) control systems, drones transmit environmental data collected on-site to a control station, which processes the information and issues operational directives to accomplish reconnaissance tasks. As depicted in \figurename~\ref{fig:scenario}, these scenarios illustrate a structured closed-loop framework consisting of four essential stages: ``Information updating'', ``service information consumption'', ``task forwarding'', and ``task execution''. The cyclic process of transmitting resource or environmental data and executing subsequent tasks is widespread across various networking scenarios.

From a more general perspective, emerging networking paradigms such as CFN, also known as Computing Power Networks (CPN) \cite{RN2652, RN3, RN1449}, terminal devices frequently offload locally generated computational tasks to remote service nodes (edge servers) for execution. To assess the availability of a service node, servers periodically or conditionally transmit their current resource status to access points (APs) when significant changes occur or when coordination is necessary. Subsequently, the AP decides whether to offload the client's task based on this information.

\begin{figure*}[htbp]
    \centering
    \includegraphics[width=\linewidth]{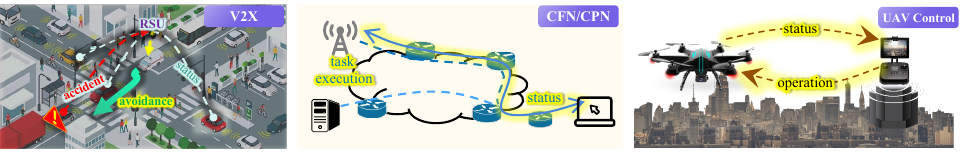}
    \caption{Scenarios of compute-aware task offloading}
    \label{fig:scenario}
\end{figure*}

It can be seen that, the timeliness and accuracy of service information (i.e, resource status) in these systems directly impact the likelihood of successful task execution. Incorrect or outdated reports can cause the AP to misjudge the availability or capability of service nodes, leading to failed task redirection or inefficient offloading decisions. In contrast, excessively frequent updates can lead to information congestion and increased communication resource consumption, particularly in scenarios requiring the transmission of large volumes of monitoring data. To address this trade-off, the research community has introduced theoretical frameworks such as of Information (AoI) \cite{RN1372} and and Value of Information (VoI) \cite{RN2760} to assess and optimize information update mechanisms.

AoI is a performance metric designed to quantify the timeliness of information. Unlike traditional end-to-end delay, AoI focuses on the time elapsed between the generation of an update and its reception at the destination. More precisely, in a system where information updates are continuously generated, the AoI at a given moment is defined as the time that has passed since the most recently received update was generated. When a new packet arrives, AoI is reset to the generation-to-arrival delay of that packet; otherwise, AoI increases linearly over time. This metric intuitively reflects the freshness of information: A smaller AoI indicates fresher data, whereas a large AoI implies that the information has become stale. In addition, extended variants such as AoP (Age of Processing) \cite{RN1181} and QAoI (Query Age of Information) \cite{RN1437} have been proposed. AoP captures delays in processing, while QAoI specifically reflects the age of the information at the time it is queried or consumed by the decision-making node, offering a more comprehensive view of latency and relevance in complex systems.

In contrast, VoI evaluates the importance of information from a semantic and functional perspective. In essence, it quantifies the degree to which a piece of information enhances decision making or the utility of the system. For example, in network control systems, the ``value'' of the information can be interpreted as the amount by which it reduces environmental uncertainty \cite{RN2761}. In this context, VoI serves as a measure of usefulness: If an update significantly improves decision accuracy or system performance, it is considered valuable; otherwise, its contribution may be negligible. Unlike AoI, which emphasizes temporal aspects, VoI focuses on semantic significance and practical impact. 

However, in CFN scenarios, especially those that involve edge-side decision making, AoI alone may not sufficiently capture resource dynamics. For example, if the system’s state remains stable for an extended period, older information might still be relevant. Conversely, during rapid state fluctuations, the value of information may decay quickly over time, making AoI the dominant factor. This importance becomes particularly pronounced when the service supports concurrent access, making the capacity information embedded within updates especially critical, especially in the context of rapidly evolving virtualization technologies. To the best of our knowledge, the CFN literature still lacks comprehensive methods that simultaneously account for both the content and freshness of information in optimizing task offloading decisions.

To address this gap, we investigate a single-source, single-destination CFN system and propose a unified metric, the Age-and-Value-Aware (\metricname) information metric, to guide the update strategy of remote service information. Specifically, we define the service node’s capacity state as information, and consider the content of the information, its dissemination process, and client access behavior to model the system as a Markov Decision Process (MDP). We design a novel reward function based on both AoI and VoI, and apply reinforcement learning to optimize the update policy using the Proximal Policy Optimization (PPO) algorithm \cite{schulman2017proximal}. Extensive experimental results demonstrate that AVA achieves a 92.04\% reduction in update rate compared to QAoI in deterministic access patterns and a 98. 34\% reduction compared to AoP in uniform access scenarios, all while maintaining the accuracy of the decision. 
This paper makes the following key contributions:
\begin{itemize}
    \item We focus on single-source, single-destination system and propose a novel metric, \metricname, which integrates both information freshness and semantic consistency to guide service information updates in CFN.
    \item We formulate the update decision process as an MDP and design a customized reward function that incorporates our AVA metric. Meanwhile, we develop and train a reinforcement learning agent using PPO to learn optimal update strategies without requiring an explicit model of the system. 
    \item Extensive simulations demonstrate that AVA achieves up to 98.34\% reduction in update rates compared to existing methods while maintaining decision accuracy.
\end{itemize}

The remainder of this paper is structured as follows: Section \ref{sect:related-work} reviews related work on information effectiveness metrics and update optimization strategies. Section \ref{sect:system-mode-and-problem} describes the system model, key definitions, and introduces the proposed AVA metric. Section \ref{sect:mdp} formulates the optimization problem as a Markov Decision Process (MDP) and propose the PPO solution. Section \ref{sect:analysis} presents theoretical analysis. Section \ref{sect:experiments} reports simulation settings and experimental results. Finally, Section \ref{sect:conclusion} concludes the paper.

\section{Related Work} \label{sect:related-work}
We review related studies from two perspectives: Metrics for evaluating information effectiveness and methods for optimizing update decisions.

\subsection{Information Effectiveness Metrics}
In the context of monitoring information effectiveness, the AoI has been widely adopted to quantify the time elapsed from when a status update is generated at the source node to when it is received at the destination (e.g., the edge node or AP). Metrics such as Average AoI and Peak AoI (PAoI) are commonly used to assess the freshness of information over time \cite{RN1372}. Recent work has expanded the scope of AoI to incorporate queueing effects, energy consumption, transmission loss, and dynamic network delays, aiming to ensure that the freshest possible information is available at any moment \cite{RN1372, RN39}.

Several studies have emphasized that information becomes most valuable only after it is analyzed or processed \cite{RN7, RN1, RN6}. Moltafet et al. \cite{RN48} extended the AoI model by accounting for bidirectional delays in the network, including both information dispatch and result feedback. Pan et al. \cite{RN83} considered unreliable links and bidirectional latency, using a non-decreasing function of AoI (i.e., estimation error) to evaluate information effectiveness. In our earlier work \cite{RN2759}, we proposed a three-stage information effectiveness metric called TPAoI, which integrates the processes of status update, edge-side information consumption, and task dispatch. This metric effectively reduces update frequency while preserving freshness at the decision point.

However, the monitored system state often evolves as a nonlinear stochastic process \cite{RN51, RN34, RN14, RN2},  making it difficult for time-based metrics alone to capture the availability of remote resources. This can result in either inaccurate decision-making or increased update frequency \cite{RN51, RN34, RN14, RN2}. As a result, numerous AoI variants have been proposed, such as Value of Information (VoI), Age of Incorrect Information, Age of Changed Information, and Age of View \cite{RN1372, RN73, RN37, RN25}.

Other studies have further incorporated client-side consumption behavior and task execution dynamics into their evaluations \cite{RN40, RN1181, RN26, RN37}. Yin and Chiariotti \cite{RN38, RN1437} emphasized that information attains its maximum value only when it is used for decision-making. Consequently, developing effectiveness metrics that consider client access behavior, information freshness, and semantic content has become a prevailing research direction. Nevertheless, how to construct such metrics specifically for status updates in CFN remains an open research question.

\subsection{Update Decision Optimization}
Quantifying information effectiveness ultimately serves to support the development of better update strategies that improve long-term system performance. Depending on whether the system model is known, the methods can be categorized as model-based or model-free.

Model-based methods require an accurate probabilistic model of the environment’s state transitions. For instance, Chiariotti et al. \cite{RN1437} constructed a comprehensive model involving AoI, transmission reliability, energy cost, and user behavior. In our prior work \cite{RN53}, we analytically derived the relationship between update intervals and task execution performance under concurrent multi-user access. However, such model-based approaches are often infeasible in complex or dynamic environments.

By contrast, model-free methods, such as reinforcement learning, are particularly suitable for environments that are difficult to model explicitly. These methods focus on designing reward functions and maximizing cumulative rewards through interaction with the environment. For example, Holm et al. \cite{RN25} applied a DQN-based approach, minimizing the mean-squared error between observed and actual values to enhance information accuracy. The superiority of model-free methods in handling complex and partially observable systems has been widely validated.

Despite these advances, a key challenge remains: how to effectively integrate information freshness, the dynamic capacity of service nodes, and user access behavior in order to design update strategies that not only reduce communication overhead but also ensure high decision accuracy in CFN environments. Addressing this challenge is the central motivation of our work.

\section{System Model and Problem Formulation} \label{sect:system-mode-and-problem}
We consider a time‐slotted system operating under a centralized CFN architecture, a structure that has become increasingly popular in recent years \cite{RN1784}, focusing on a single‐source, single‐destination system configuration, where each time slot has duration \(u\). Within this framework, a decision node continuously observes the system state and determines whether a service information update should be triggered; this decision is then synchronized to the service node.

The service node subsequently broadcasts its current service information (e.g., the number of remaining idle threads) to the AP, which leverages this data to make task offloading decisions on behalf of the connected clients.

To distinguish between the service information and the observed system state, we first introduce the following two definitions:
\begin{definition}[System state]
The \emph{system state} refers to the global and time-varying condition of the entire CFN system. It encompasses the real-time status of service nodes, the arrival patterns of client tasks, network transmission characteristics, and overall resource availability.
\end{definition}
The system state is observed by the decision node (which may be a service node or a centralized controller), and it serves as the basis for determining when to update service information and how to make task offloading decisions.
 
\begin{definition}[Service information]
\emph{Service information} refers to the resource status message broadcast by a service node, in this paper, most commonly represented by its residual processing capacity (e.g., thread count $c_t$). 
\end{definition}
Due to transmission latency $\tau\!>\!0$, the AP receives a delayed observation $\hat{c}_t := c_{t-\tau}$, which may not reflect the server’s current state $c_t$.

Some important symbols and explanations can be found in \tablename~\ref{tbl:symbols}. 

\begin{table}[htbp]
    \caption{List of important symbols}
    \label{tbl:symbols}
    \centering
    \renewcommand{\arraystretch}{1.2}
    \begin{tabular}{c|p{7cm}}
        \hline
        \textbf{Symbol} & \textbf{Meaning / Description}  \\ \hline \hline 
        $C_{\text{MAX}}$      & Maximum concurrent threads a service node can process \\ \hline
        $c_t$                 & Actual available thread count at the service node at time $t$  \\ \hline
        $\hat{c}_t$           & Thread count observed by the AP at time $t$ \\ \hline
        $w_{j}$               & Workload size of the $j^{th}$ task  \\ \hline
        $V_s$                 & Per-thread service rate of the server \\ \hline
        $\tau$                & Network transmission delay of a service information update \\ \hline
        $t_{q,i}$             & Arrival instant of the $i^{th}$ client query at the AP  \\ \hline
        $\Delta(t)$           & AoI at time $t$  \\ \hline
        $\Upsilon_i$          & $\Delta(t_{q,i})$ sampled at query $i$, i.e., QAoI  \\ \hline
        $\phi(\Upsilon_i)$    & Normalized QAoI \\ \hline
        $\mathcal{R}(c_t,\hat{c}_t)$ & Information-consistency metric between server and AP  \\ \hline
        $\delta(t)$           & AVA metric $\phi(\Upsilon_i)\cdot\mathcal{R}(c_t,\hat{c}_t)$\\ \hline
    \end{tabular}
\end{table}
\subsection{The Stages of Service Information Consumption}\label{sect:service-information-stages}

The overall process of service information consumption is illustrated in \figurename~\ref{fig:system-model}.  It consists of the following stages: (1) \textit{Service Information Updating}: Each service node broadcasts its latest service information (e.g., the currently available concurrent processing capacity) to the APs via the network. (2) \textit{Stochastic User Accessing}: Clients randomly access the APs and submit task requests. (3) \textit{Task Forwarding}: The AP, upon receiving a client request, evaluates the currently available service information and determines whether to forward the task to a service node for execution. (4) \textit{Task Running}: Upon receiving the forwarded task, the service node executes the task using its available resources and subsequently returns the result to the client.

\begin{figure}[htbp]
    \centering
    \includegraphics[width=\linewidth]{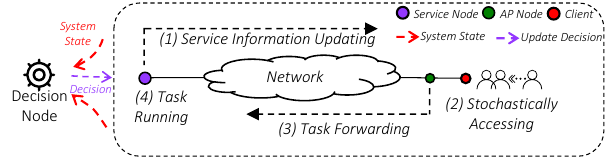}
    \caption{Illustration of the information-consumption stages}
    \label{fig:system-model}
\end{figure}

In systems where service nodes serve as decision nodes, they directly observe and collect the system state.
In architectures featuring a centralized controller (e.g., Software-Defined Networking, SDN), the controller acts as the decision node, performing system state observation and coordination, typically realized through telemetry mechanisms.

For simplicity, we assume that the system does not experience sudden anomalies or traffic bursts, such as network congestion or equipment failures.

\subsection{Key Observations}
State‑of‑the‑art update schemes seek to deliver information to the AP precisely at (or immediately before) each client access, thereby minimizing the AoI at query instants (e.g., QAoI \cite{RN1437}). Although these techniques succeed in preserving freshness with minimal update cost, they overlook the fact that modern service nodes often employ multi-threading, virtualization, and other concurrency mechanisms to handle multiple requests in parallel. Under such conditions, strictly minimizing AoI at access times leads to redundant updates that waste resources without improving the quality of the decision.
\begin{figure}[htbp]
  \centering
  \includegraphics[width=\linewidth]{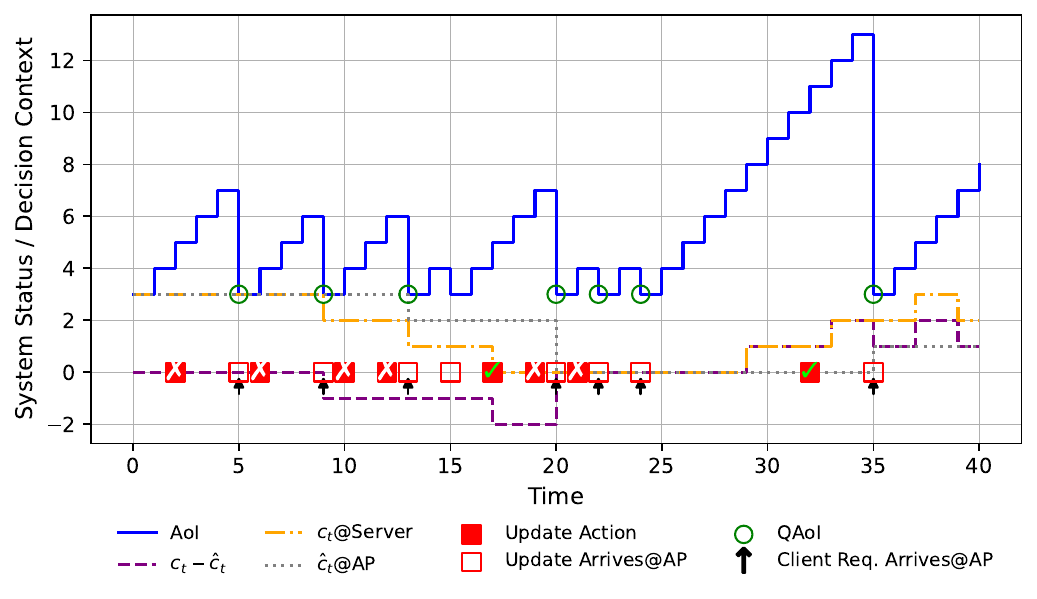}
  \caption{Timeline of system status and key decision events.}
  \label{fig:evaluation–example}
\end{figure}

To clarify this phenomenon, \figurename~\ref{fig:evaluation–example} presents a concrete scenario. In this example, we assume the server’s maximum concurrent capacity $C_{\text{MAX}}$ is 3. Both the service information update to the AP and the request forwarding to the server incur a delay of 3 time slots, and the server requires 20 time slots to process each request.
Under these assumptions, we plot the real-time AoI, the service information update decisions (Update Action), the arrival of information in the AP (Update Arrives@AP), the true capacity of the server ($c_t$@Server), the AP’s observed capacity ($\hat c_t$@AP), and the instant when the client request reaches the AP (Client Req. Arrives@AP). We also annotate the optimal QAoI values, i.e., making the update arrives at AP exactly when the client request arrives at AP.

For clarity, we mark each decision‑affecting update action with ``\ding{52}'' and each redundant update with ``\ding{53}''. It is evident that an update is efficient only when the true server capacity ($c_t$@Server) and the AP’s observed capacity ($\hat{c}_t$@AP) differ significantly (i.e., when the difference ($c_t - \hat{c}_t$) is maximal). When this difference is small, although the AP’s information does not exactly match the server’s state, decision accuracy remains unaffected.

\subsection{System Entities and Communication}

The aforementioned system involves the following key roles.

\subsubsection*{Service Node (Server)}
The service node supports a maximum processing capacity of $C_{\text{MAX}}$ concurrent threads. At time $t$, the number of available threads is denoted as $c_t \in \{0,1,\ldots,C_{\text{MAX}}\}$. Each thread operates at a fixed service rate (or speed) $V_s$. When the decision node triggers an update, the service node transmits its current service information, including $c_t$, to the AP over the network.

When the service node receives a task (with a workload size of $w_j$) forwarded from the AP, it checks the availability of processing threads. If $c_t > 0$, it allocates one thread to the task and processes it at a rate determined by $w_j / V_s$. Upon task completion, the thread is released, and $c_t$ is incremented by one, up to the maximum capacity $C_{\text{MAX}}$.

\subsubsection*{Access Point (AP)} At time $t-\tau$, the service node sends the service information $c_{t-\tau}$ to the AP, which receives it after a transmission delay $\tau$, thus observing $\hat{c}_{t}$. When a client task request arrives at the AP, the AP evaluates the received service information $\hat{c}_{t}$ to determine whether to forward the request to the service node. The forwarding strategy adopts the FIFO (first-in, first-out) principle.

Upon receiving an update packet, in practice, the AP immediately sends an acknowledgment (ACK) back to the service node (or decision node), enabling the decision node to accurately track the observed service information at any given time.
\subsubsection*{Client} The client generates task requests that arrive at the AP according to a stochastic process, such as one with uniformly distributed inter-arrival times. Each client task $j$ is characterized by a workload size $w_j$ and also follows a random variable.

\subsubsection*{Network Model} To ensure information freshness and prevent backlog, the forwarding nodes in the network maintain a queue of size one; newly received service information updates are immediately forwarded to the AP.
The transmission delay $\tau$ is assumed to depend on the size of the transmitted information. For simplicity, we model the transmission delay as an exponential random variable, $\tau \sim \text{Exp}(\lambda)$, where $\lambda$ is the network service rate (i.e., the number of service messages that can be transmitted per time slot). For the delay of client request forwarding, we also assume it costs $\lambda$.
We quantize the continuous delay into integer slots via $\tilde\tau=\lceil\tau/u\rceil$ to match our slotted implementation.

\subsection{The \metricname~Metric} \label{sect:metric}
Due to the lack of real-time synchronization between the AP and the Server, when a client task arrives at the AP, the following scenarios may occur:\\
\begin{align}
\text{Right Decision} &= 
\begin{cases}
    \hat{c}_t > 0 \land c_t > 0, & \text{true positive} \\
    \hat{c}_t = 0 \land c_t = 0, & \text{true negative}
\end{cases},
\end{align}
and 
\begin{align}
\text{Error Decision} &= 
\begin{cases}
    \hat{c}_t > 0 \land c_t = 0, \text{false positive} \\
    \hat{c}_t = 0 \land c_t > 0, \text{false negative}
\end{cases}.
\end{align}

In other words, a decision is considered erroneous if and only if the AP’s judgment contradicts the true state of the service node, resulting in either an incorrect acceptance or an incorrect rejection in real-time. For server state prediction, i.e., $c_{t+\tau}$, we leave it in future work.

Based on the above decision cases, we define the value (or content)-based service information effectiveness metric as:
\begin{align}
\mathcal{R}(c_t, \hat{c}_t) =
\begin{cases}
1, & \text{if } c_t, \hat{c}_t \leq \epsilon \\
2 \cdot \frac{\min(c_t, \hat{c}_t)}{\max(c_t, \hat{c}_t) + \epsilon} - 1, & \text{otherwise}
\end{cases}
\end{align}
where $\epsilon$ is a small positive constant to keep the result from meaningless. 
The value of $\mathcal{R}(c_t, \hat{c}_t)$ lies within $[-1, 1]$ and monotonically increases with $\frac{\min(c_t, \hat{c}_t)}{\max(c_t, \hat{c}_t) + \epsilon}$.

Thus, the more consistent the service node’s information $c_t$ and the AP’s observed information $\hat{c}_t$, the higher the metric; conversely, greater discrepancies result in a lower score.

It is important to emphasize that, unlike traditional AoII (Age of Incorrect Information) metrics \cite{RN2, RN34}, $\mathcal{R}(c_t, \hat{c}_t)$ positively rewards synchronization trends (improving consistency) and negatively evaluates diverging trends (increasing inconsistency).

Beyond content synchronization, ensuring information timeliness is equally critical, given that client task arrivals are random.
The conventional metric is AoI, as in \eqref{eq:aoi}.
The AoI at time t is defined as:
\begin{align}\label{eq:aoi}
\Delta(t) = t - \max_{i: t_{i} \leq t} t_{i},
\end{align}
where $t_{i}$ denotes the reception time of the $i^{th}$ successfully received update; $\Delta(t)$ captures the time elapsed since the latest successfully received update by time $t$. 
To minimize unnecessary updates, a desirable approach is to align information updates as closely as possible with the moments just preceding client task arrivals, a concept captured by QAoI \cite{RN1437} (as in \eqref{eq:aoq}).
Let the query instants be denoted as $t_{q,1}, t_{q,2}, \ldots$. Then, the QAoI for the $i^{th}$ query is defined as:
\begin{align}\label{eq:aoq}
    \Upsilon_i = \Delta(t_{q,i}).
\end{align}

That is, QAoI $\Upsilon_i$ samples the AoI at each query arrival time $t_{q,i}$.
Accordingly, to restrict its size range, we introduce a normalized QAoI adjustment function:
\begin{align}
\phi(\Upsilon_i) = 1 - e^{-k  \Upsilon_i},
\end{align}
where $k$ is a positive constant (in our experiments, we choose $k$ = 0.02), and $\phi(\Upsilon_i) \in [0, 1)$. 

However, in periods without client arrivals, the system (or server) tends to remain stable. In these cases, larger $\phi(\Upsilon_i)$ (triggered by the latest query) are preferable to avoid repeated updates. 
The interpretation of $\delta(t)$ is as follows:
\begin{align}\label{eq:metric}
    \delta(t) = \phi(\Upsilon_i) \cdot \mathcal{R}(c_t, \hat{c}_t)
\end{align}
By \eqref{eq:metric}, we have the following observations:
\begin{itemize}
    \item When the service information discrepancy is large ($\mathcal{R}(c_t, \hat{c}_t) \to -1)$, a larger $\phi(\Upsilon_i)$ amplifies the negative signal, prompting more frequent updates.
    \item When the service information is closely synchronized ($\mathcal{R}(c_t, \hat{c}_t) \to 1$), a larger $\phi(\Upsilon_i)$ amplifies the positive signal, encouraging reduced update frequency and promoting communication efficiency.
\end{itemize}

\subsection{Problem Formulation}
In order to balance information freshness, content consistency, and update overhead, based on the effectiveness metric $\delta(t)$, we formulate the system optimization objective as maximizing the long-term average system performance over $T$ discrete time slots:
\begin{align}
    \max \ \liminf_{T \to \infty} \frac{1}{T} \mathbb{E} \left[ \sum_{t=0}^{T-1} \left( \delta(t) - c_u a_t \right) \right],
\end{align}
where $a_t\in\{0,1\}$ triggers an update costing $c_u>0$.

\section{MDP Formulation and Solutions}\label{sect:mdp}
We model the service information update and task forwarding system as a discrete-time MDP, characterized by the tuple $\langle\mathcal{S},\mathcal{A},\mathcal{P},\mathcal{R},\gamma\rangle$.

\subsection{State Space $\mathcal{S}$}
To accurately capture the dynamic behavior of the system and support effective decision-making, we carefully design the state space to incorporate critical aspects of resource availability, information freshness, and client access patterns. 
    
At each time step t, the system state is defined as:
    \begin{align}
        s_t = \left( c_t,\hat{c}_t,\Upsilon_i, \tau_t, \iota_t\ \right),
    \end{align}
where $\iota_t \in \mathbb{N}$ is the time elapsed since the last client arrival.

This design is motivated by the following considerations. The pair ($c_t,\hat{c}_t$) provides a comprehensive view of the resource state and the potential mismatch between actual and observed conditions, which is critical for task forwarding decisions. The inclusion of $\Upsilon_i$ allows the agent to assess the freshness of the service information precisely at the moment when it is actually needed, rather than passively tracking the information age over time. 

Moreover, by leveraging the subsequent effectiveness metric $\delta(t)$, the agent can jointly consider both the timeliness and the content consistency of the service information, thus enabling more intelligent update decisions that maintain information quality while significantly reducing unnecessary update frequency.

To reflect the serialized transmission nature observed in real systems, we introduce a scalar variable \( \tau_t \in \{0, 1, 2, \dots, \tau_{max}\} \) representing the remaining delay before the in-transit update arrives at the AP. This variable governs the update timing as follows:
\begin{itemize}
    \item If \( \tau_t > 0 \), the update is still in transit. We decrement \( \tau_{t+1} \leftarrow \tau_t - 1 \). When \( \tau_t = 0 \), the AP updates \( \hat{c}_{t+1} \) with the content of the arriving update.
    \item If \( \tau_t = 0 \) and the decision \( a_t = 1 \), the server initiates a new update carrying the current \( c_t \), and sets \( \tau_{t+1} \leftarrow \tau_{max} \).
\end{itemize}

Finally, $\iota_t$ captures the temporal dynamics of client access behavior, enabling the system to anticipate future load conditions based on recent access intervals. 
Thus, the state space is: $\mathcal{S} = \mathbb{R}_{\geq 0} \times \{0,1,\ldots,C_{\text{MAX}}\}^2 \times \tau_{max} \times \mathbb{N}$.

\subsection{Action Space $\mathcal{A}$}
At each time step $t$, the agent selects an action $a_t$ from the discrete action space:
\begin{align}
a_t \in \mathcal{A} = \{0, 1\}
\end{align}
where:\begin{itemize}
    \item $a_t = 0$ means no service information update is performed;
    \item $a_t = 1$ means a service information update is triggered, incurring a constant cost $c_u$.
\end{itemize}

\subsection{Transition Dynamics $\mathcal{P}$}
Given $s_t$ and $a_t$, the conditional distribution of $s_{t+1}$ obeys:
\begin{itemize}
\item \emph{Server Availability}
    \begin{align}
      c_{t+1} =
        \begin{cases}
            \max(0, c_t - 1), & \text{if a request is forwarded} \\
                              & \quad \text{and accepted}, \\
            \min(C_{\text{MAX}}, c_t + 1), & \text{if a task finished}, \\
            c_t, & \text{otherwise}.
        \end{cases}
    \end{align}
\item \emph{Observed Capacity (on AP)} 

\begin{align}
\hat{c}_{t+1}=
    \begin{cases}
        c^{\text{received}}, & \text{if  } \tau_t = 1, \\
        \hat{c}_t, & \text{otherwise}
    \end{cases},
\end{align}
where $c^{\text{received}}$ is the newly received service information at AP.
\begin{align}
\tau_{t+1} =
\begin{cases}
    \tau_{max}, & \text{if } \tau_t = 0 \text{ and } a_t = 1 \\
    \tau_t - 1, & \text{if } \tau_t > 0 \\
    0, & \text{otherwise}
\end{cases}
\end{align}

\item \emph{QAoI} $\Upsilon_{i+1}=\Delta(t_{q,i+1})$ whenever a request arrives; else $\Upsilon_{i+1}=0$.
\item \emph{Inter-arrival timer} $\iota_{t+1}=0$ on an arrival request, otherwise $\iota_t+1$.
\end{itemize}

\subsection{Reward $\mathcal{R}$ and Objective}
The immediate reward at time $t$ is defined as:
\begin{align}
r_t = \delta(t) - c_u \times a_t,
\end{align}
with $\delta(t)$ given by~\eqref{eq:metric}.

The overall optimization objective is to find a stationary policy $\pi:\mathcal{S}\!\to\!\mathcal{A}$ maximizing the expected discounted return:
\begin{align}
J(\pi)=\mathbb{E}_{\pi}\Bigl[\sum_{t=0}^{\infty}\gamma^{\,t}\,r_t\Bigr],\qquad 0<\gamma<1.
\end{align}

Consequently, the joint evolution of the system state $s_t$ under the control actions $a_t$ forms a discrete-time MDP.

\subsection{Solutions with PPO Methods}

We learn the update decisions with PPO via Stable-Baselines3\footnote{https://stable-baselines3.readthedocs.io}. A small actor–critic network takes the state ($c_t,\hat c_t,\Upsilon_i,\iota_t$) and outputs the probability of “update” versus “no-update”. After collecting short trajectories and computing the advantage signal, we adjust the actor by maximizing the \emph{clipped} surrogate objective, with a small entropy bonus for exploration.

\noindent \textit{Remark on $\tau_{t}$:}
In practical CFN deployments, the transmission delay $\tau$ may be stochastic and vary over time due to changing network conditions. Moreover, the maximum delay $\tau_{\max}$ is often unknown or difficult to estimate a priori, making it challenging to explicitly model the full delay process within the system state.

To address this challenge, we adopt a model-free reinforcement learning framework. Specifically, we employ the Proximal Policy Optimization (PPO) algorithm, which enables the agent to learn an effective update strategy through direct interaction with the environment—without requiring an explicit model of delay dynamics or a bounded state space related to delay variables such as $\tau_t$.

Since PPO is a model-free reinforcement learning method, the agent can learn to exploit this discrepancy without explicitly modeling the stochastic delay. Empirical results validate that the learned policy achieves substantial reductions in update frequency while preserving or even improving decision accuracy.

\section{Theoretical Analysis}\label{sect:analysis}

This section establishes the analytical foundations of the proposed method, proving that the optimization problem is well posed and that a stationary optimal policy exists.  We also show that the reward landscape satisfies the regularity conditions required by policy-gradient methods such as PPO.

\subsection{Fundamental Properties of the \metricname~ Metric}

\begin{proposition}[Boundedness and Normalization]\label{prop:bounded}
For all $t\in\mathbb{N}$,
$
-1\le\delta(t)
     =\phi(\Upsilon_i)\,\mathcal{R}(c_t,\hat{c}_t)
     \le 1.
$
\end{proposition}

\begin{proof}
$\mathcal{R}(c_t,\hat{c}_t)\in[-1,1]$ by construction;  
$\phi(\Upsilon_i) = 1-e^{-k\Upsilon_i} \in [0,1)$ for $k>0$ and $\Upsilon_i\ge0$.  
Therefore their product remains in $[-1,1)$.
\end{proof}

\begin{proposition}[Monotonicity]\label{prop:mono}
Fix $c_t,\hat{c}_t$.  Then $\delta(t)$ is strictly increasing in $\Upsilon_i$ if $\mathcal{R}(c_t,\hat{c}_t)>0$, and strictly decreasing if $\mathcal{R}(c_t,\hat{c}_t)<0$.
\end{proposition}

\begin{proof}
$\partial\phi/\partial\Upsilon_i=k e^{-k\Upsilon_i}>0$.  
Hence the sign of $\partial\delta(t)/\partial\Upsilon_i$ equals the sign of $\mathcal{R}(c_t,\hat{c}_t)$.
\end{proof}
\noindent \textit{Remark.} If $\mathcal{R}(c_t,\hat{c}_t) = 0$, then $\delta(t)=0$ for any $\phi(\Upsilon_i)$,  reflecting that no update adds value when the server and AP information coincide.

\subsection{Regularity of the Controlled Process}

\begin{assumption}\label{ass:ergodic}
\begin{enumerate}[label=(\alph*)]
\item The server capacity $C_{\max}$ is finite.
\item Inter-arrival times are i.i.d.\ with finite mean.
\item Task sizes $w_j$ are bounded above by $r_{\max}<\infty$.
\end{enumerate}
\end{assumption}

\begin{lemma}[Positive Recurrence]\label{lem:recurrence}
Given Assumption~\ref{ass:ergodic}, the Markov chain $\{s_t\}_{t\ge0}$ induced by any stationary policy $\pi$ is
\emph{irreducible}, \emph{aperiodic}, and \emph{positive recurrent}.
\end{lemma}

\begin{proof}[Proof idea]
\emph{Step 1 (finite drift).}  
Because the server capacity is bounded ($0\!\le c_t,\hat c_t\le C_{\max}$) and every task size satisfies $w_j\!\le w_{\max}\!<\!\infty$, the queue of unfinished tasks can never grow without bound: Within at most $C_{\max}$ service completions the system must return to a state with $c_t=C_{\max}$.

\emph{Step 2 (reachability).}  
The link delay $\tau\sim\mathrm{Exp}(\lambda)$ has full support on $\mathbb{N}$, hence there is a strictly positive probability of observing \emph{every} delay length.  
Combining this with the renewal (i.i.d.) arrival process implies that, from any state, there exists a sequence of events—arrival, possible forwarding, service completion, and update delay—leading to any other state in finitely many steps with non-zero probability.  
Thus the chain is irreducible.  Because some self-transition has positive probability, it is also aperiodic.

\emph{Step 3 (return time).}
Starting from any state, the expected time to hit the ``all-idle'' state $(c_t=C_{\max},\hat c_t=C_{\max},\Upsilon_i=0,\tau_t=0, \iota_t=0)$ is finite:  
Arrivals are light-tailed, service times are bounded, and the delay distribution is exponential.  
Therefore the chain’s expected return time to that state is finite, establishing positive recurrence.
\end{proof}

\subsection{Existence of a Stationary Optimal Policy}

\begin{theorem}\label{thm:optimal}
Given Assumption~\ref{ass:ergodic} and Lemma~\ref{lem:recurrence}, there exists a deterministic stationary policy $\pi^\star$ that maximize the expected discounted return
$
J(\pi)=\mathbb{E}_{\pi}\!\bigl[\sum_{t=0}^{\infty}\gamma^{\,t}r_t\bigr]
$
for every $\gamma\in(0,1)$.
\end{theorem}

\begin{proof}
Lemma~\ref{lem:recurrence} guarantees that the reward process is bounded (Prop.~\ref{prop:bounded}) and the state space countable but $\sigma$-compact.  
Standard results in Markov Decision Theory (e.g. Puterman, 1994, Thm.~6.2.10 \cite{puterman2014markov}) then imply the existence of an optimal deterministic stationary policy. \qedhere
\end{proof}

\section{Experiments}\label{sect:experiments}
\subsection{Simulation Environment}
The experiments were conducted on a MacBook Pro equipped with an Apple M4 Pro processor and 24 GB of RAM. To better align with practical scenarios, we implement a simulation environment based on ns-3 version 3.43\footnote{https://www.nsnam.org}, ns3-gym\cite{ns3gym}, and the Stable-Baselines3 library, which runs inside a Docker container to facilitate fast and portable deployment. Inside the container, gcc 11.4.0 and Python 3.10.12 were installed to support the simulation and reinforcement learning frameworks. \figurename~\ref{fig:simulation-enviroment} illustrates the architecture of the environment.

\begin{figure}[htbp]
    \centering
    \includegraphics[width=0.9\linewidth]{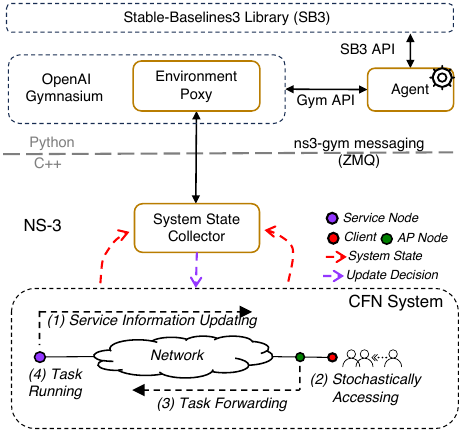}
    \caption{The implemented simulation environment}
    \label{fig:simulation-enviroment}
\end{figure}

At the bottom layer, the CFN System is modeled in ns-3, consisting of three types of nodes: A service node, an AP node, and a client. The client stochastically access the network (step 2), triggering task forwarding (step 3) toward service nodes, where tasks are executed (step 4). Meanwhile, service information is updated according to the Agent decision (step 1) within the network.

The System State Collector gathers real-time system state from the CFN system and transmits it to the Environment Proxy via the ns3-gym messaging interface (ZMQ protocol). The Environment Proxy, built on OpenAI Gymnasium\footnote{https://gymnasium.farama.org}, acts as a bridge between the simulation and the reinforcement learning agent.

At the top layer, an Agent implemented with Stable-Baselines3 (SB3) interacts with the environment through the Gym API, making update decisions based on the observed system state. These decisions are fed back to ns-3 to dynamically control the system’s behavior.

\subsection{Experiment Setup}
\subsubsection{Baselines}
We compare our proposed method against three representative baseline strategies, each aiming to minimize a different notion of information freshness: AoI-minimized update, AoP-minimized update~\cite{RN1181}, and QAoI-minimized update~\cite{RN1437}.

To ensure a fair comparison, we adopt a consistent cost model across all methods: Every time a service information update is performed, a fixed update cost $c_u$ is incurred. At the same time, we implemented their PPO version to deal with dynamic environmental changes.
\subsubsection{Evaluation Metrics} 
To assess the efficiency of information updates and their impact on resource consumption, we measure the update rate, which quantifies the frequency of service information updates relative to the total number of time slots:
\begin{align*}
\text{Update Rate} = \frac{\text{Total Updates}}{\text{Total Time Slots}}.
\end{align*}

We also evaluate the effectiveness of service information in supporting task-forwarding decisions. A key performance metric is the decision accuracy, defined as the proportion of correct decisions (both true positives, TP, and true negatives, TN) over the total number of client requests:
\begin{align*}
    \text{Accuracy} = \frac{\text{TP} + \text{TN}}{\text{Total Requests}}.
\end{align*}

These two metrics provide a comprehensive evaluation of both the quality of the decision and the efficiency of the communication of the different update strategies.
\subsubsection{Simulation Settings}
To emulate dynamic client access behaviors, we consider two distinct access models. The first is a \textit{deterministic periodic} access pattern, which reflects scenarios such as satellite communication systems~\cite{RN1437}, where client requests arrive at regular intervals.
Specifically, we vary the request interval across the set \{0.0167, 0.02, 0.025, \dots, 0.1\} seconds, corresponding to arrival rates of 60, 50, 40, \dots, 10 requests per second. 
The second is a uniformly distributed random access model, where the inter-arrival times of client requests are sampled from a uniform distribution. Here, the \textit{minimum} interval again takes values in \{0.0167, 0.02, 0.025, \dots, 0.1\}, and the maximum is set to 1.5$\times$\textit{minimum}. These two models allow us to evaluate the robustness and adaptability of the update strategies under both structured and stochastic client behaviors. We simulate an IEEE 802.11a WLAN in infrastructure mode when the client accesses the AP node.

For the task size $w_j$ configuration, we adopt a uniform distribution model, where each task size is randomly sampled from a uniform distribution with parameters: $w_j \sim \mathcal{U}(40, 45)$.
It is important to note that this work does not focus on queuing or scheduling dynamics on the server side. Therefore, variations in task size are not the primary factor under investigation and are set to a narrow range to maintain consistency while introducing slight randomness.

To introduce dynamic variability in the delay in sending the information from the service, we model the delay as the sum of a fixed-base delay and a stochastic component. Specifically, based on our network topology, each forward delay is composed of a constant of $4$ time slots plus a random delay sampled from an exponential distribution with rate parameter $\lambda$. 
This setup reflects realistic network conditions where the baseline transmission time is predictable, but random fluctuations may occur due to varying network load or link conditions. The actual duration corresponding to one time slot can be flexibly adjusted according to practical system requirements, such as milliseconds for edge computing scenarios or seconds for satellite communication systems. To align with ns-3, we set one time slot equal to one millisecond.

We implement PPO using the Stable‑Baselines3 library’s built‑in PPO class, \tablename~\ref{tbl:rl-parameters} shows the parameters used in the experiments. 
We train our PPO agent with episodes of 2000 time‐steps each, ensuring that the full pattern of client request arrivals is observed before an update. Gradients are accumulated over two episodes and network parameters are updated every four episodes. Training runs for a total of 5$\times 10^{\text{5}}$ steps. All other parameters remain at the Stable‑Baselines3 defaults. During testing, we evaluate over 100 seconds and compute the average value of each metric.
\begin{table}[htbp]
\caption{Parameters for PPO training}
\centering
\begin{tabular}{c|c}
    \hline
    \textbf{Parameter} & \textbf{Value}\\
    \hline
    \hline
    information update cost $c_u$&  0.5\\
    \hline
    discounted factor $\gamma$&  0.995\\
    \hline
    learning rate $\eta$& 3$\times 10^{\text{-4}}$  \\
    \hline
    total training steps & 5$\times 10^{\text{5}}$ \\
    \hline
    episode size& 2000 \\
     \hline
    mini-batch size $\Xi$& 4000 \\
    \hline
\end{tabular}
\label{tbl:rl-parameters}
\end{table}
\subsection{Results}
\subsubsection{Efficiency Evaluation}
To validate the efficiency of the AVA metric, we varied the server’s maximum capacity $C_{\text{MAX}}$ and compared the update rates of the four methods, as shown in \figurename~\ref{fig:update-deterministic} and \figurename~\ref{fig:update-uniform}. In the scenario with a uniform distribution user access pattern ($min = 0.05, max = 0.075$, \figurename~\ref{fig:update-i}), corresponding to a request rate of approximately $13–20$ times per second, the proposed AVA method demonstrated significant efficiency gains. When compared to the baseline AoP, which required an average of 34.09 updates/s, AVA reduced this to only 2.05 updates/s, achieving a 93.84\% reduction in communication overhead. In comparison, AoI and QAoI yielded 16.56 and 2.49 updates/s respectively, indicating that AVA also performs competitively or better than other baseline methods under the same workload. Overall, our approach clearly outperforms the other three in reducing update rate.

Specifically, under deterministic periodic access (Figs. \ref{fig:update-a} - \ref{fig:update-f}), both AVA and QAoI exhibit decreasing update rates as the client inter‑arrival interval increases. For shorter intervals (Figs. \ref{fig:update-a}–\ref{fig:update-d}), AVA consistently yields a lower update rate than QAoI, demonstrating its superior adaptability to high request frequencies. Under a Deterministic user access pattern with an inter-request interval of 0.025 seconds (i.e., 40 requests per second), and with a server concurrency capacity of $C_{\text{MAX}} = 4$, the AVA method dramatically reduced update frequency. While the baseline method QAoI required an average of 13.56 updates per second, AVA achieved the same task with only 1.08 updates per second, resulting in a 92.04\% reduction in update cost.

Moreover, as the server’s capacity $C_{\text{MAX}}$ grows, the AVA-driven update rate declines steadily, indicating that AVA effectively captures the interplay between server capability and client behaviors.
Under the more stochastic uniform access model (\figurename~\ref{fig:update-uniform}), the advantage of AVA is similarly pronounced. In the scenario with a uniformly distributed user request pattern ($min = 0.0333 s, max = 0,05 s$, \figurename~\ref{fig:acc-j}) and a moderate server capacity of $C_{\text{MAX}} = 2$, the proposed AVA method significantly reduced the update rate. Compared to the baseline AoP, which required 67.62 updates per second, AVA reduced this to 1.12 updates per second, achieving a 98.34\% reduction in communication overhead.

We also observe that when client arrival rates are low, specifically in deterministic scenarios (Figs.~\ref{fig:update-e} and \ref{fig:update-f}) and uniform scenarios (Figs.~\ref{fig:update-k} and \ref{fig:update-l}), AVA and QAoI converge to nearly identical update rates. In these scenarios, the task duration roughly matches or smaller than the request interval, causing the server to process requests serially. In other words, AVA attains a minimal update frequency comparable to QAoI.

\begin{figure}[htbp]
    \centering
    \begin{subfigure}[b]{0.48\linewidth}
        \centering
        \includegraphics[width=\linewidth]{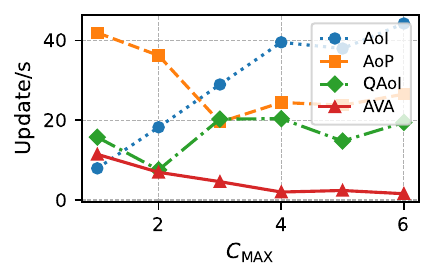}
        \caption{Request interval = 0.0167}
        \label{fig:update-a}
    \end{subfigure}
    \hfill
    \begin{subfigure}[b]{0.48\linewidth}
        \centering
        \includegraphics[width=\linewidth]{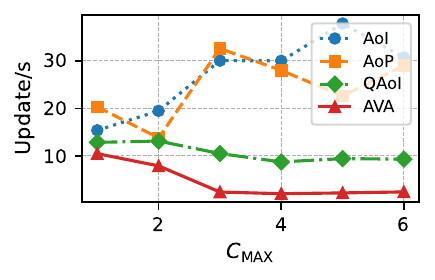}
        \caption{Request interval = 0.02}
        \label{fig:update-b}
    \end{subfigure}
    \begin{subfigure}[b]{0.48\linewidth}
        \centering
        \includegraphics[width=\linewidth]{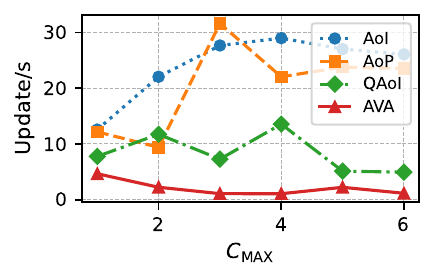}
        \caption{Request interval = 0.025}
        \label{fig:update-c}
    \end{subfigure}
    \hfill
    \begin{subfigure}[b]{0.48\linewidth}
        \centering
        \includegraphics[width=\linewidth]{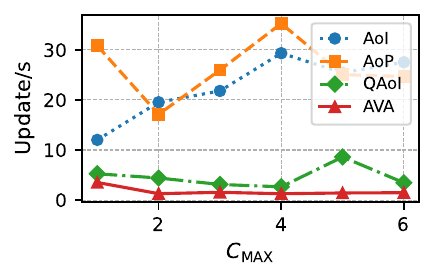}
        \caption{Request interval = 0.0333}
        \label{fig:update-d}
    \end{subfigure}
    \begin{subfigure}[b]{0.48\linewidth}
        \centering
        \includegraphics[width=\linewidth]{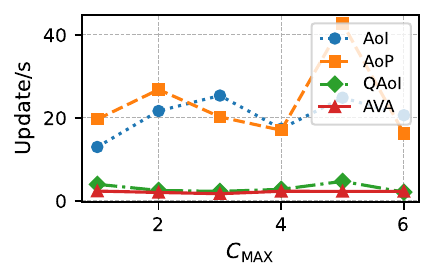}
        \caption{Request interval = 0.05}
        \label{fig:update-e}
    \end{subfigure}
    \hfill
    \begin{subfigure}[b]{0.48\linewidth}
        \centering
        \includegraphics[width=\linewidth]{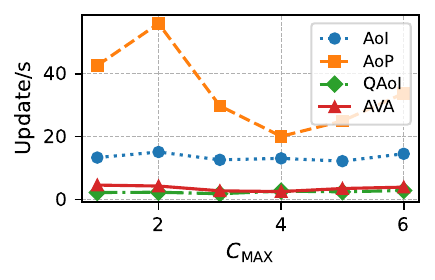}
        \caption{Request interval = 0.1}
        \label{fig:update-f}
    \end{subfigure}
    \caption{Update rate under deterministic distributed user access patterns}
    \label{fig:update-deterministic}
\end{figure}

\begin{figure}[htbp]
    \centering
    \begin{subfigure}[b]{0.48\linewidth}
        \centering
        \includegraphics[width=\linewidth]{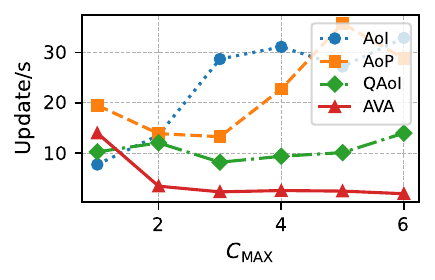}
        \caption{min=0.0167, max=0.025}
        \label{fig:update-g}
    \end{subfigure}
    \hfill
    \begin{subfigure}[b]{0.48\linewidth}
        \centering
        \includegraphics[width=\linewidth]{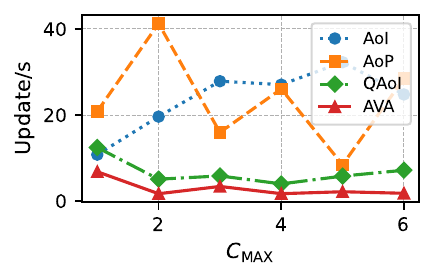}
        \caption{min=0.02, max=0.03}
        \label{fig:update-h}
    \end{subfigure}
    \begin{subfigure}[b]{0.48\linewidth}
        \centering
        \includegraphics[width=\linewidth]{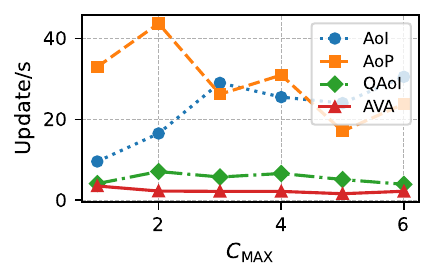}
        \caption{min=0.025, max=0.0375}
        \label{fig:update-i}
    \end{subfigure}
    \hfill
    \begin{subfigure}[b]{0.48\linewidth}
        \centering
        \includegraphics[width=\linewidth]{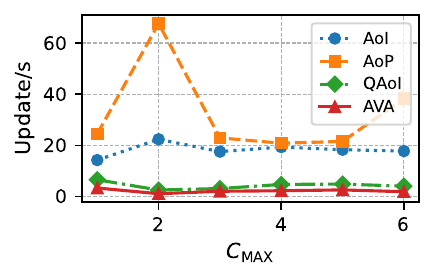}
        \caption{min=0.0333, max=0.05}
        \label{fig:update-j}
    \end{subfigure}
    \begin{subfigure}[b]{0.48\linewidth}
        \centering
        \includegraphics[width=\linewidth]{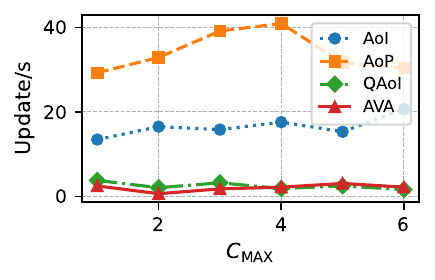}
        \caption{min=0.05, max=0.075}
        \label{fig:update-k}
    \end{subfigure}
    \hfill
    \begin{subfigure}[b]{0.48\linewidth}
        \centering
        \includegraphics[width=\linewidth]{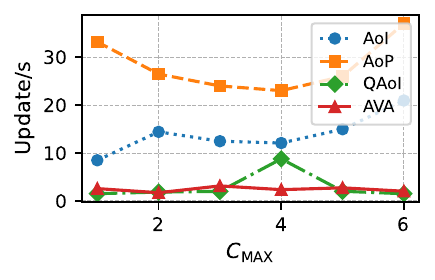}
        \caption{min=0.1, max=0.15}
        \label{fig:update-l}
    \end{subfigure}
    \caption{Update rate under uniform distributed user access patterns}
    \label{fig:update-uniform}
\end{figure}

\subsubsection{Effectiveness Evaluation}
To validate AVA’s enhancement of information effectiveness, we evaluate the accuracy as shown in Figs. \ref{fig:acc-deterministic} and \ref{fig:acc-uniform}. Overall, AVA sustains strong accuracy even at very low update rates, thereby reducing energy consumption and confirming the efficiency of our approach.

Under identical access patterns (e.g., \figurename~\ref{fig:acc-g}), increasing $C_{\text{MAX}}$ improves accuracy for all methods, reflecting reduced resource contention and more reliable edge‑side information.

\begin{figure}[htbp]
    \centering
    \begin{subfigure}[b]{0.48\linewidth}
        \centering
        \includegraphics[width=\linewidth]{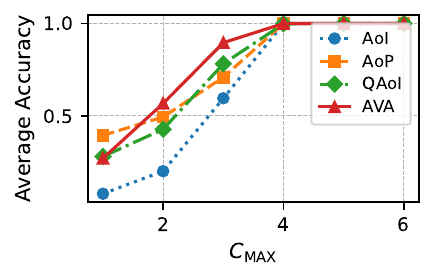}
        \caption{Request interval = 0.0167}
        \label{fig:acc-a}
    \end{subfigure}
    \hfill
    \begin{subfigure}[b]{0.48\linewidth}
        \centering
        \includegraphics[width=\linewidth]{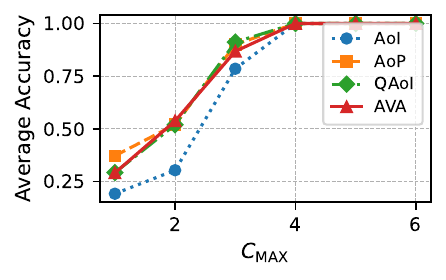}
        \caption{Request interval = 0.02}
        \label{fig:acc-b}
    \end{subfigure}
    \begin{subfigure}[b]{0.48\linewidth}
        \centering
        \includegraphics[width=\linewidth]{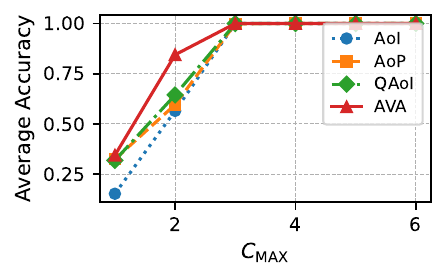}
        \caption{Request interval = 0.025}
        \label{fig:acc-c}
    \end{subfigure}
    \hfill
    \begin{subfigure}[b]{0.48\linewidth}
        \centering
        \includegraphics[width=\linewidth]{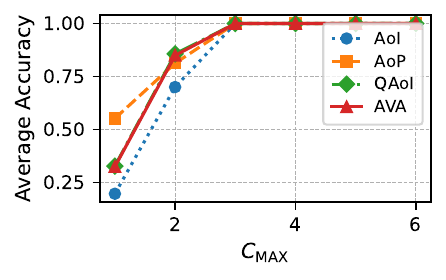}
        \caption{Request interval = 0.0333}
        \label{fig:acc-d}
    \end{subfigure}
    \begin{subfigure}[b]{0.48\linewidth}
        \centering
        \includegraphics[width=\linewidth]{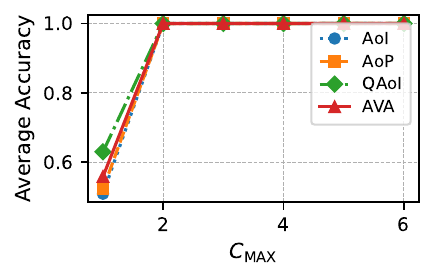}
        \caption{Request interval = 0.05}
        \label{fig:acc-e}
    \end{subfigure}
    \hfill
    \begin{subfigure}[b]{0.48\linewidth}
        \centering
        \includegraphics[width=\linewidth]{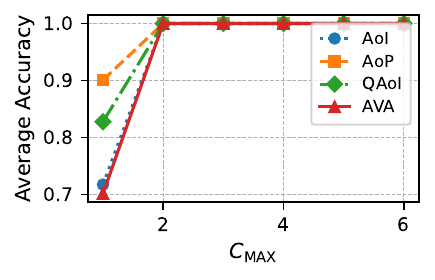}
        \caption{Request interval = 0.1}
        \label{fig:acc-f}
    \end{subfigure}
    \caption{Average accuracy under deterministic distributed user access patterns}
    \label{fig:acc-deterministic}
\end{figure}
\begin{figure}[htbp]
    \centering
    
    \begin{subfigure}[b]{0.48\linewidth}
        \centering
        \includegraphics[width=\linewidth]{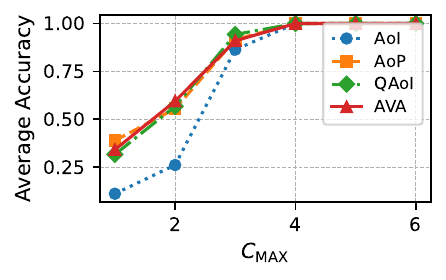}
        \caption{min=0.0167, max=0.025}
        \label{fig:acc-g}
    \end{subfigure}
    \hfill
    \begin{subfigure}[b]{0.48\linewidth}
        \centering
        \includegraphics[width=\linewidth]{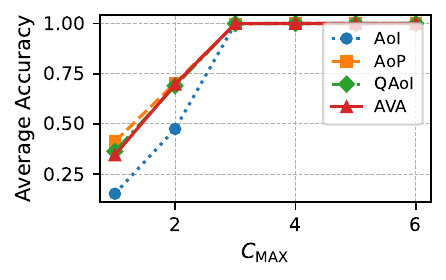}
        \caption{min=0.02, max=0.03}
        \label{fig:acc-h}
    \end{subfigure}
    \begin{subfigure}[b]{0.48\linewidth}
        \centering
        \includegraphics[width=\linewidth]{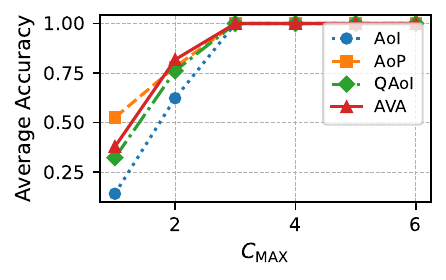}
        \caption{min=0.025, max=0.0375}
        \label{fig:acc-i}
    \end{subfigure}
    \hfill
    \begin{subfigure}[b]{0.48\linewidth}
        \centering
        \includegraphics[width=\linewidth]{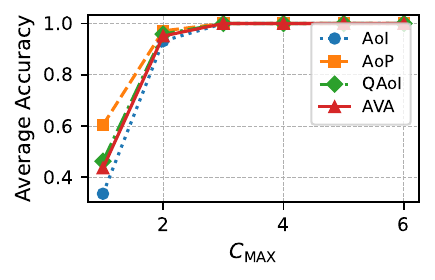}
        \caption{min=0.0333, max=0.05}
        \label{fig:acc-j}
    \end{subfigure}
    \begin{subfigure}[b]{0.48\linewidth}
        \centering
        \includegraphics[width=\linewidth]{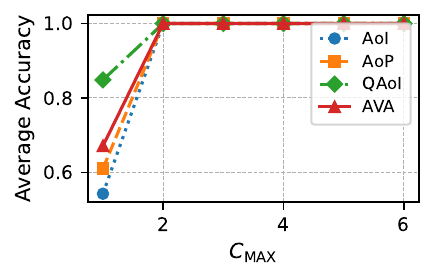}
        \caption{min=0.05, max=0.075}
        \label{fig:acc-k}
    \end{subfigure}
    \hfill
    \begin{subfigure}[b]{0.48\linewidth}
        \centering
        \includegraphics[width=\linewidth]{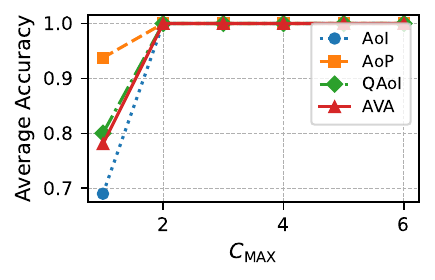}
        \caption{min=0.1, max=0.15}
        \label{fig:acc-l}
    \end{subfigure}
    \caption{Average accuracy under uniform distributed user access patterns}
    \label{fig:acc-uniform}
\end{figure}
Moreover, we measured the average AoI of the service information at the AP under both access distributions, as shown in Figs. \ref{fig:aoi-deterministic} and \ref{fig:aoi-uniform}. It is important to note that AoI captures only freshness, not effectiveness. Both AVA and QAoI maintain high AoI levels, reflecting their low update rates.

Furthermore, comparing the methods reveals that AVA yields even higher AoI, especially as service capacity increases. Combined with the accuracy results in Figs. \ref{fig:acc-deterministic} and \ref{fig:acc-uniform}, this demonstrates that a larger AoI does not imply ineffectiveness, thereby validating our underlying hypothesis.
\begin{figure}[htbp]
    \centering
    \begin{subfigure}[b]{0.48\linewidth}
        \centering
        \includegraphics[width=\linewidth]{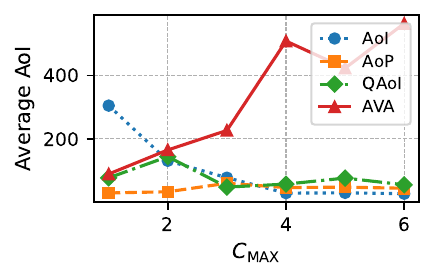}
        \caption{Request interval = 0.0167}
        \label{fig:aoi-a}
    \end{subfigure}
    \hfill
    \begin{subfigure}[b]{0.48\linewidth}
        \centering
        \includegraphics[width=\linewidth]{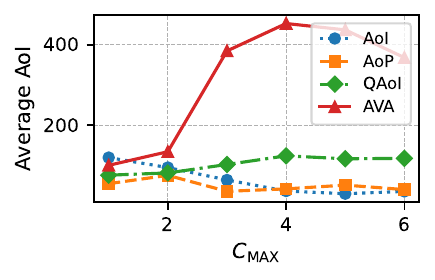}
        \caption{Request interval = 0.02}
        \label{fig:aoi-b}
    \end{subfigure}
    \begin{subfigure}[b]{0.48\linewidth}
        \centering
        \includegraphics[width=\linewidth]{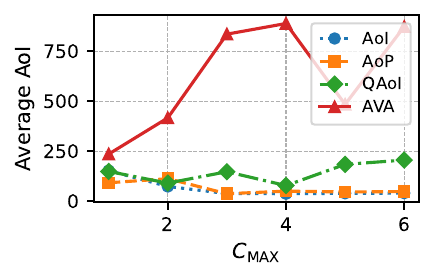}
        \caption{Request interval = 0.025}
        \label{fig:aoi-c}
    \end{subfigure}
    \hfill
    \begin{subfigure}[b]{0.48\linewidth}
        \centering
        \includegraphics[width=\linewidth]{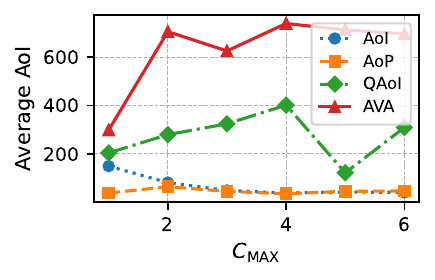}
        \caption{Request interval = 0.0333}
        \label{fig:aoi-d}
    \end{subfigure}
    \begin{subfigure}[b]{0.48\linewidth}
        \centering
        \includegraphics[width=\linewidth]{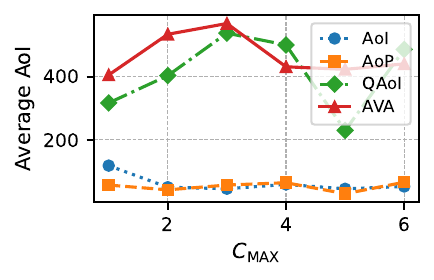}
        \caption{Request interval = 0.05}
        \label{fig:aoi-e}
    \end{subfigure}
    \hfill
    \begin{subfigure}[b]{0.48\linewidth}
        \centering
        \includegraphics[width=\linewidth]{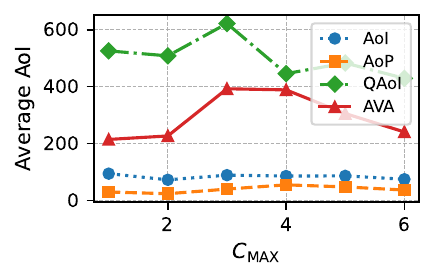}
        \caption{Request interval = 0.1}
        \label{fig:aoi-f}
    \end{subfigure}
    \caption{Average AoI under deterministic distributed user access patterns}
    \label{fig:aoi-deterministic}
\end{figure}

\begin{figure}[htbp]
    \centering
    \begin{subfigure}[b]{0.48\linewidth}
        \centering
        \includegraphics[width=\linewidth]{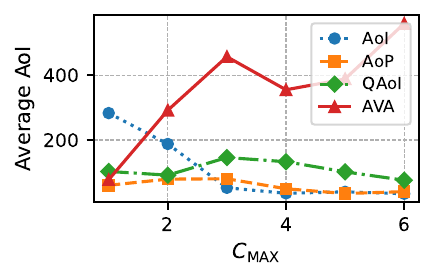}
        \caption{min=0.0167, max=0.025}
        \label{fig:aoi-g}
    \end{subfigure}
    \hfill
    \begin{subfigure}[b]{0.48\linewidth}
        \centering
        \includegraphics[width=\linewidth]{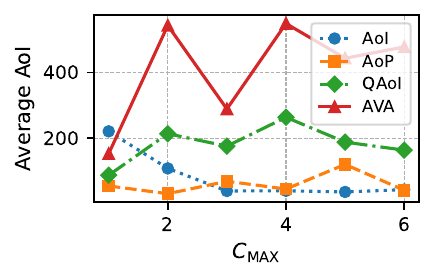}
        \caption{min=0.02, max=0.03}
        \label{fig:aoi-h}
    \end{subfigure}
    \begin{subfigure}[b]{0.48\linewidth}
        \centering
        \includegraphics[width=\linewidth]{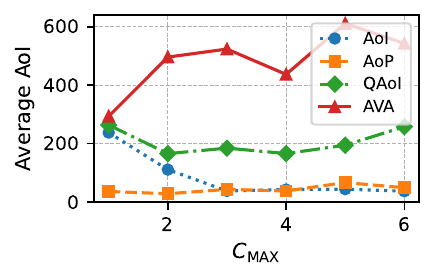}
        \caption{min=0.025, max=0.0375}
        \label{fig:aoi-i}
    \end{subfigure}
    \hfill
    \begin{subfigure}[b]{0.48\linewidth}
        \centering
        \includegraphics[width=\linewidth]{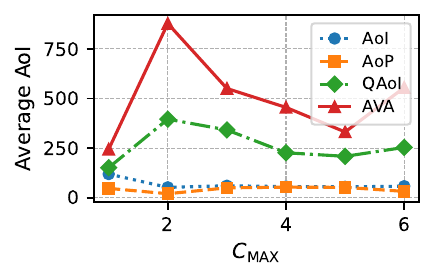}
        \caption{min=0.0333, max=0.05}
        \label{fig:aoi-j}
    \end{subfigure}
    \begin{subfigure}[b]{0.48\linewidth}
        \centering
        \includegraphics[width=\linewidth]{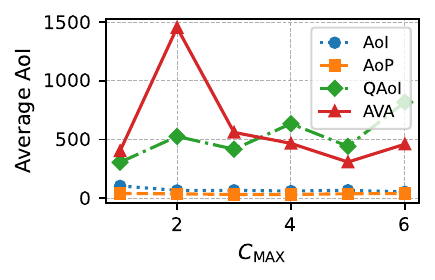}
        \caption{min=0.05, max=0.075}
        \label{fig:aoi-k}
    \end{subfigure}
    \hfill
    \begin{subfigure}[b]{0.48\linewidth}
        \centering
        \includegraphics[width=\linewidth]{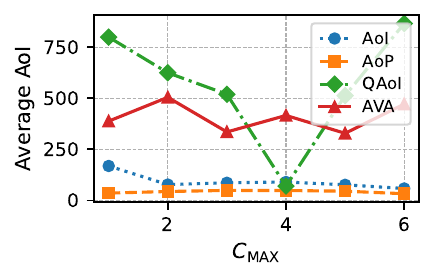}
        \caption{min=0.1, max=0.15}
        \label{fig:aoi-l}
    \end{subfigure}
    \caption{Average AoI under uniform distributed user access patterns}
    \label{fig:aoi-uniform}
\end{figure}

\section{Conclusion} \label{sect:conclusion}
In this work, we introduced the Age‑and‑Value‑Aware (AVA) metric to jointly capture information freshness and semantic consistency for status updates in compute‑first networking. By formulating the update decision as an MDP and optimizing it via PPO, AVA dynamically balances update cost against decision accuracy. Our extensive experiments under both deterministic and uniform access patterns demonstrate that AVA achieves substantially lower update rates than AoI, AoP, and QAoI while sustaining equal or higher task‑forwarding accuracy. Finally, our analysis of average AoI confirms that increased age need not imply ineffectiveness when updates are value‑aware. Together, these results validate AVA as an efficient, robust strategy for status refresh in compute‑first networking.
Future work will extend AVA to multi‐source CFN and incorporate queuing dynamics for more realistic deployments.
\balance
\bibliographystyle{plain} 
\bibliography{references}

\end{document}